%% file: main.tex
\documentclass[12pt]{article}
\usepackage[margin=1in, centering]{geometry}
\usepackage{mathtools}
\usepackage{amsthm}
\usepackage{dsfont}
\usepackage{amssymb}
\usepackage{libertine}
\usepackage{newtxmath}
\usepackage[table]{xcolor}
\usepackage{esvect}
\usepackage{subcaption}
\usepackage{tikz}
\usetikzlibrary{angles,quotes}
\usepackage[normalem]{ulem}
\definecolor{TSUYUKUSA}{RGB}{46, 169, 223}
\definecolor{KURENAI}{RGB}{203, 27, 69}
\usepackage[pdfstartview=FitH,pdfpagemode=UseNone,colorlinks=true,citecolor=KURENAI,linkcolor=TSUYUKUSA,backref=page,linktocpage=true]{hyperref}
\usepackage[capitalise]{cleveref}
\usepackage[nottoc]{tocbibind}
\usepackage{appendix}
\usepackage{makecell}

\input{macro}

\allowdisplaybreaks

\begin{document}


\title{A Pseudorandom Generator for Functions of Low-Degree Polynomial Threshold Functions}

\author{Penghui Yao\thanks{\scriptsize State Key Laboratory for Novel Software Technology, New Cornerstone Science Laboratory, Nanjing University, Nanjing 210023, China. Email:
    \texttt{phyao1985@gmail.com}.}~\thanks{\scriptsize Hefei National
    Laboratory, Hefei 230088, China.} %
  \and Mingnan Zhao\thanks{\scriptsize State Key Laboratory for Novel Software Technology, New Cornerstone Science Laboratory, Nanjing University, Nanjing 210023, China. Email:
    \texttt{mingnanzh@gmail.com}.}  }
\date{April 15, 2025}
\maketitle

\thispagestyle{empty}
\begin{abstract}
  \input{abstract}

\end{abstract}


\setcounter{page}{1}

\section{Introduction}
\label{sec:intro}
\input{intro}


\paragraph{Acknowledgment.} PY and MZ were supported by
National Natural Science Foundation of China (Grant No. 62332009 and 12347104),
Innovation Program for Quantum Science and Technology (Grant No. 2021ZD0302901),
NSFC/RGC Joint Research Scheme (Grant No. 12461160276),
Natural Science Foundation of Jiangsu Province (Grant No. BK20243060),
and the New Cornerstone Science Foundation.

\section{Preliminary}
\label{sec:prelim}
\input{preliminary}

\section{Fooling the Functions of PTFs via Bounded Independence}
\label{sec:fool}
\input{fool.tex}


\section{Discretization}
\label{sec:discrete}
\input{discrete.tex}

\newpage
\bibliographystyle{alpha}
\bibliography{ref}

\newpage
\begin{appendices}
\addappheadtotoc
\addtocontents{toc}{\protect\setcounter{tocdepth}{1}}

\input{app}

\end{appendices}

\end{document}

%% file: macro.tex
\newtheorem{theorem}{Theorem}[section]

\newtheorem{fact}[theorem]{Fact}
\newtheorem{lemma}[theorem]{Lemma}

\newtheorem{definition}[theorem]{Definition}

\newcommand{\fig}[1]{\hyperref[fig:#1]{Figure~\ref*{fig:#1}}}
\newcommand{\eq}[1]{\hyperref[eq:#1]{(\ref*{eq:#1})}}
\newcommand{\lem}[1]{\hyperref[lem:#1]{Lemma~\ref*{lem:#1}}}
\newcommand{\thm}[1]{\hyperref[thm:#1]{Theorem~\ref*{thm:#1}}}
\newcommand{\defi}[1]{\hyperref[def:#1]{Definition~\ref*{def:#1}}}
\newcommand{\app}[1]{\hyperref[sec:#1]{Appendix~\ref*{sec:#1}}}
\newcommand{\fct}[1]{\hyperref[fact:#1]{Fact~\ref*{fact:#1}}}
\newcommand{\sect}[1]{\hyperref[sec:#1]{Section~\ref*{sec:#1}}}
\newcommand{\subsec}[1]{\hyperref[subsec:#1]{Subsection~\ref*{subsec:#1}}}
\newcommand{\itm}[2]{\hyperref[itm:#1]{#2}}
\newcommand{\clm}[1]{\hyperref[clm:#1]{Claim~\ref*{clm:#1}}}
\newcommand{\rmk}[1]{\hyperref[rmk:#1]{Remark~\ref*{rmk:#1}}}
\newcommand{\tbl}[1]{\hyperref[tbl:#1]{Table~\ref*{tbl:#1}}}
\makeatletter
\newtheorem*{rep@theorem}{\rep@title}
\newcommand{\newreptheorem}[2]{%
\newenvironment{rep#1}[1]{%
 \def\rep@title{#2 \ref*{##1}}%
 \begin{rep@theorem}}%
 {\end{rep@theorem}}}
\makeatother

\newreptheorem{fact}{Fact}

\definecolor{lightcyan}{RGB}{0.88,1,1}
\definecolor{darkgreen}{RGB}{0, 128, 0}
\definecolor{darkblue}{RGB}{0, 0, 128}

\newcommand{\N}{\mathbb{N}}  %
\newcommand{\R}{\mathbb{R}} %
\DeclareMathOperator*{\E}{\mathbb{E}}  %
\newcommand{\A}{\mathcal{A}}  %
\newcommand{\F}{\mathcal{F}} %
\newcommand{\NN}{\mathcal{N}} %
\newcommand{\EE}{\mathcal{E}}  %
\newcommand{\bit}[1]{\{0,1\}^{#1}} %
\newcommand{\expec}[1]{\E\!\Br{#1}} %
\newcommand{\expect}[2]{\E_{\substack{#1}}\!\Br{#2}} %

\newcommand{\sign}[1]{\mathrm{sign}\!\br{#1}}
\newcommand{\grad}[2]{\nabla^{#1}{#2}}

\newcommand{\br}[1]{\left(#1\right)} %
\newcommand{\Br}[1]{\left[#1\right]} %
\newcommand{\st}[1]{\left\{#1\right\}} %
\newcommand{\abs}[1]{\left|#1 \right|} %
\newcommand{\norm}[1]{\left\lVert #1 \right\rVert} %
\newcommand{\poly}[1]{\mathrm{poly}\!\br{#1}} %
\newcommand{\id}{\ensuremath{\mathds{1}}} %
\usepackage[normalem]{ulem}

%% file: abstract.tex
Developing explicit \emph{pseudorandom generators} (PRGs) for prominent categories of Boolean functions is a key focus in computational complexity theory.
In this paper, we investigate the PRGs against the
functions of degree-$d$
\emph{polynomial threshold functions} (PTFs)
over Gaussian space.
Our main result is an explicit construction of PRG
with seed length $\poly{k,d,1/\epsilon}\cdot\log n$
that can fool \emph{any} function of $k$
degree-$d$ PTFs with probability at least $1-\varepsilon$.
More specifically, we show that the summation of
$L$ independent $R$-moment-matching Gaussian vectors
$\epsilon$-fools functions of $k$ degree-$d$ PTFs,
where $L=\poly{ k, d, \frac{1}{\epsilon}}$ and $R = O({\log \frac{kd}{\epsilon}})$.
The PRG is then obtained by applying an appropriate discretization
to Gaussian vectors with bounded independence.

%% file: intro.tex
In computational complexity theory,
derandomization is a powerful technique that aims to
reduce randomness in algorithms without sacrificing efficiency or accuracy.
A versatile approach for derandomization is to design explicit \emph{pseudorandom generators} (PRGs) for notable families of Boolean functions.
A PRG for a family of Boolean functions is able to consume few random bits
and produce a distribution over high-dimensional vectors,
which is indistinguishable from a target distribution, such as the uniform distribution over Boolean cube, by any function in the family. In this paper, we concern ourselves with the Gaussian distribution over $\R^n$.
Formally,
\begin{definition}
	Let $\F\subseteq\st{f:\R^n \to \{0,1\}}$ be a family of Boolean functions.
	A function $G:\{0,1\}^r \to \R^n$ is a pseudorandom generator for $\F$ with error $\epsilon$ over Gaussian distribution $\NN\br{0,1}^n$ if for any $f\in\F$,
	\[
		\abs{ \expect{s\sim_u\st{0,1}^r}{f(G(s))} 
		 - \expect{x\sim\NN\br{0,1}^n}{f(x)}		
		} \leq \epsilon \enspace. 
	\]
	We call $r$ the seed length of $G$. We also say $G$ $\epsilon$-fools $\F$ over the Gaussian distribution.
\end{definition}

There has been a considerable amount of research
developing PRGs for various Boolean function families,
including halfspaces, polynomial threshold functions and intersections of halfspaces.
Let $\mathrm{sign}:\R\to\st{0,1}$ be the function such that $\sign{x}=1$ iff $x\geq0$.
A \emph{halfspace} is a Boolean function of the form $f(x) = \sign{a_1x_1+\cdots+a_nx_n-b}$ for some $a_1,\cdots,a_n,b\in\R$.
Halfspaces are a fundamental class of Boolean functions which have found significant applications in machine learning, complexity theory, theory of approximation and more.
A very successful series of work produced PRGs that $\epsilon$-fools halfspaces with seed length poly-logarithmic in $n$ and $\epsilon^{-1}$ over both Boolean space \cite{Ser06, DGJ+10, MZ13, GKM18} and Gaussian space \cite{KM15}.
\emph{Polynomial threshold functions} (PTFs) are functions of the form $f(x) = \sign{p(x)}$ where $p$ is a polynomial. We call $f$ is a degree-$d$ PTF if $p$ is a degree-$d$ polynomial. PTFs are natural generalization for halfspaces since a halfspace is a degree-$1$ PTF.
An explicit PRG that $\epsilon$-fools PTFs over Boolean space has been achieved with seed length
$ (d/\epsilon)^{O(d)}\cdot\log n $ \cite{MZ13}.
As for Gaussian space, a sequence of work \cite{DKN10, Kan11a, Kan11b, Kan12, MZ13, Kan14, Kan15, OST20, KM22} succeeds in giving a PRG with seed length polynomial in $d$, $\epsilon^{-1}$ and $\log n$ \cite{OST20, KM22}.
Another extension for halfspaces is \emph{intersections} of $k$ halfspaces
which are polytopes with $k$ facets.
A line of work \cite{GOWZ10, HKM13, ST17, CDS19, OST22} results in PRGs with seed length polynomial in $\log k$, $\log n$ and $1/\epsilon$ over Boolean space \cite{OST22} and over Gaussian space \cite{CDS19}.

%

Considering the prosperity of PRGs for these functions families, we commence designing PRGs for
\emph{functions of degree-$d$ polynomial threshold functions}.
\begin{definition}
	We say a function $F:\R^n\to\{0,1\}$ is a function of $k$ degree-$d$ PTFs if there exist $k$ polynomials $p_1,\dots,p_k:\R^n \to \R$ of degree $d$ and a Boolean function $f:\bit{k}\to\st{0,1}$ such that
	\[
		F(x) = f\!\br{\sign{p_1\!(x)},\dots,\sign{p_k\!(x)}}\enspace .
	\]
\end{definition}
\noindent This family consumes all three function families we discussed above. For example, it includes intersections of halfspaces by setting $d=1$ and $f(x)=x_1\cdots x_k$.
The research on PRGs for functions of PTFs is driven by several motivations beyond its fundamental role in derandomization tasks.
For instance, the collection of satisfying assignments of an intersection of $k$ degree-$2$ PTFs corresponds to the feasible solutions set of an $\{0,1\}$-integer quadratic programing \cite{NW06} with $k$ constraints.
The investigation into the structure of these sets has been a central focus of extensive research in areas including learning theory, counting, optimization, and combinatorics.

In this work, we consider building explicit PRGs for functions of degree-$d$ PTFs over Gaussian space. Before presenting our main result, we briefly revisit relevant prior work on fooling functions of halfspaces.

\begin{table}[t]
	\centering
	\caption{Related Work on PRGs for Intersections of PTFs}
	\label{tbl:related}
	\renewcommand{\arraystretch}{1.8}
	\newcolumntype{C}[1]{>{\centering\arraybackslash}p{#1}}
	
	\begin{tabular}{C{2.5cm}cc}
	\hline
	\noalign{\vskip 0mm}
	\hline
	\noalign{\vskip 0mm}
	\hline
	Reference & Function Family  & Seed length \\
	\hline
	\cite{GOWZ10}   & Monotone functions of $k$ halfspaces & $O((k\log (k/\epsilon) + \log n)\cdot \log (k/\epsilon))$ \\
	\hline
	\cite{HKM13}   & Intersections of $k$ $\delta$-regular halfspaces & \makecell{$O(\log n\log k/\epsilon)$\\for $\delta \leq \epsilon^5/(\log^{8.1}\! k \log (1/\epsilon)) $} \\
	\hline
	\cite{ST17}   & Intersections of $k$ weight-$t$ halfspaces & $\poly{\log n, \log k, t, 1/\epsilon}$ \\
	\hline
	\cite{OST22}    & Intersections of $k$ halfspaces & \makecell{$\mathrm{polylog}\ m\cdot{\epsilon}^{-(2+\delta)}\cdot\log n$\\for any absolute constant $\delta\in(0,1)$}\\
	\hline
	\cite{CDS19}    & \makecell{Intersections of $k$ halfspaces\\Arbitrary functions of $k$ halfspaces} & \makecell{$O(\log n + \poly{\log k,1/\epsilon})$\\$O(\log n + \poly{ k,1/\epsilon})$}\\
	\hline
	\cite{DKN10}    & Intersections of $k$ degree-$2$ PTFs & $O(\log n\cdot\poly{k,1/\epsilon})$ \\
	\hline
	\noalign{\vskip 0mm}
	\hline
	\noalign{\vskip 0mm}
	\hline
	\end{tabular}
\end{table}

\subsection{Prior Work}

The related work is summarized in \tbl{related}.
Gopalan, O’Donnell, Wu and Zuckerman \cite{GOWZ10} constructed PRGs for \emph{monotone functions of halfspaces}.
They modified the PRG for halfspaces in \cite{MZ13} and showed the modified PRG $\epsilon$-fools any monotone function of $k$ halfspaces over a broad class of \emph{product distributions} with seed length $O((k\log (k/\epsilon) + \log n)\cdot \log (k/\epsilon))$. When $k/\epsilon \leq \log^c n$ any $c>0$, the seed length can be further improved to $O(k\log (k/\epsilon) + \log n)$.

Harsha, Klivans and Meka \cite{HKM13} considered designing PRGs for intersections of \emph{regular} halfspaces (i.e., halspaces with low influence). 
A halfspace $f(x) = \sign{a_1x_1+\cdots+a_nx_n-b}$ is \emph{$\delta$-regular} if $\sum_i a_i^4\leq \delta^2 \sum_{i} a_i^2$.
They gave an explicit PRG construction for intersections of $k$ $\delta$-regular halfspaces over \emph{proper} and \emph{hypercontractive} distributions with seed length $O(\log n\log k/\epsilon)$ when $\delta$ is no more than a threshold.
Their proof is based on developing an invariance principle for intersections of regular halfspaces via a generalization of the well-known Lindeberg method \cite{Lin22} and an anti-concentration result of polytopes in Gaussian space from \cite{KOS08}.

By extending the approach of \cite{HKM13} and combing the results on bounded independence fooling CNF formulas \cite{Baz09, Raz09}, Servedio and Tan \cite{ST17} designed an explicit PRG that $\epsilon$-fools intersections of $k$ \emph{weight}-$t$ halfspaces over Boolean space with $\poly{\log n, \log k, t, 1/\epsilon}$ seed length. A halfspace $f(x) = \sign{a_1x_1+\cdots+a_nx_n-b}$ is said to be weight-$t$ if each $a_i$ is an integer in $[-t,t]$.

As for intersections of $k$ general halfspaces, O’Donnell, Servedio and Tan \cite{OST22} gave a PRG construction over Boolean space with a polylogarithmic seed length dependence on $k$ and $n$. Their proof involves a novel invariance principle for intersections of arbitrary halfspaces and a Littlewood–Offord style anticoncentration inequality for polytopes over Boolean space.

Concurrently, Chattopadhyay, De and Servedio \cite{CDS19} proposed a simple PRG that $\epsilon$-fools intersections of $k$ general halfspaces over Gaussian space, building upon the concept of
\emph{Johnson-Lindenstrauss transform} \cite{JL86, KMN11}. The seed length is $O(\log n + \poly{\log k,1/\epsilon})$. Additionally, they show that the same PRG with seed length $O(\log n + \poly{k,1/\epsilon})$ is able to fool arbitrary functions of $k$ halfspaces.

Speaking of fooling functions of PTFs, the study by Diakonikolas, Kane and Nelson \cite{DKN10} stands out as the sole work that constructs a PRG for intersections of $k$ degree-$2$ PTFs. Their PRG is specific to degree $d\leq 2$ with a $O(\log n\cdot\poly{k,1/\epsilon})$ seed length.

\subsection{Main Result}


In this work, we investigate the PRGs fooling any function of low-degree PTFs. The main result is the following.

\begin{theorem}\label{thm:main}(Informal version of \thm{main_formal})
	There exists an explicit PRG
	$\epsilon$-fools any function of $k$ degree-$d$ PTFs over Gaussian space
	with seed length $\poly{k,d,1/\epsilon}\cdot\log n$.
\end{theorem}

The proof is inspired by the PRG proposed in \cite{Kan11b} and the work \cite{KM22}. This theorem follows from two components.

\paragraph{(1) Bounded independence fools functions of $k$ degree-$d$ PTFs.}
Consider the continuous random vector
$Y = \frac{1}{\sqrt{L}} \sum_{i=1}^{L} Y_i$
where $Y_i$ is a $R$-wise independent standard Gaussian vector of length $n$.
Every $Y_{i,j}$ is a standard Gaussian variable and for any degree-$R$ polynomial $f$,
$\expec{f(Y_i)} = \expect{y\sim\NN(0,1)^n}{f(y)}$.
We will prove that
\begin{theorem}(Informal version of \thm{cont})
	With $R = O({\log \frac{kd}{\epsilon}})$ and $L=\poly{k, d, \frac{1}{\epsilon}}$, the distribution of $Y$
	$\epsilon$-fools any function of $k$ degree-$d$ PTFs over Gaussian space.
\end{theorem}
\noindent The prior work \cite{KM22} shows that bounded independence fools a single low-degree polynomial threshold function. This generalizes their work to the case of functions of $k$ low-degree PTFs.

\paragraph{(2) Discretization of bounded independence Gaussians.}
An explicit PRG construction requires a discrete approximation to Gaussian vectors with bounded independence.
The idea is to use a finite entropy random variable $X$ to approximate $Y$.
Previous work \cite{Kan11b} uses the idea that a single Gaussian variable can be produced by two uniform random variables in $[0,1]$ through the Box–Muller transform \cite{Bm58}.
Therefore bounded independence Gaussian variables $Y_i$ can be generated by using bounded independence uniform random variables.
Then by truncating these uniform $[0,1]$ random variables to a sufficient precision, we obtain vectors $X_i$ that serve as a discrete approximation of $Y_i$.
We prove that $X$ also fools functions of $k$ degree-$d$ PTFs as long as $X$ is a good approximation to $Y$.
\begin{lemma}(Informal version of \lem{discrete})
	If $X_{i,j}$ and $Y_{i,j}$ are sufficiently close with high probability, then $X$ also fools functions of $k$ degree-$d$ PTFs.
\end{lemma}

%% file: preliminary.tex
\paragraph{Basic Notation.}
For $n\in \N$, $[n]$ denotes the set $\st{1,2,\cdots,n}$.
For $\alpha\in \R^n$ and $i\in[n]$, $\alpha_i$ denotes the $i$-th coordinate of $\alpha$, $\abs{\alpha} = \sum_{i=1}^n \abs{\alpha_i}$ and $\norm{\alpha}_{\infty} = \max_{1\leq i\leq n} \abs{\alpha_i}$.
For $\alpha, \beta\in \R^n$, $\alpha-\beta$ denotes the vector $v$ such that $v_i = \alpha_i-\beta_i$ for all $i\in[n]$, and $\alpha^{\beta} = \prod_{i=1}^n \alpha_i^{\beta_i}$.
For $\alpha\in \N^n$, $\alpha! = \prod_{i=1}^n \alpha_i!$.
When it is clear from the context, we will use both subscript and superscript as indices.
\vspace{-1em}
\paragraph{Derivatives and Multidimensional Taylor Expansion.}
For a function $f:\R^n \to \R$ and $\alpha\in \N^n$, we use $\partial^{\alpha}\!f$ to denote the partial derivative taken $\alpha_i$ times in the $i$-th coordinate and define $\norm{\grad{t}{f(x)}} = \sqrt{\sum_{\alpha\in\N^n, \abs{\alpha}=t} \br{\partial^{\alpha} f(x)}^2}$.
For $f(a,b): \R^n \times \R^n \to \R$ and $\alpha,\beta\in \N^n$, we use $\partial^{\alpha}_a\partial^{\beta}_b\!f$ to denote the partial derivative taken $\alpha_i$ times in $a_i$ and $\beta_i$ times in $b_i$.
Using these notations, one has:
\begin{theorem}[Multidimensional Taylor's Theorem]\label{thm:taylor}
	Let $d\in \N$ and $f:\R^n \to \R^n$ be a $\mathcal{C}^{d+1}$ function.
	Then for all $x,y\in\R^n$,
	\[
		f(y) = \sum_{\alpha\in\N^n, \abs{\alpha}\leq d}	\frac{\partial^{\alpha} f(x)}{\alpha!} (y-x)^{\alpha} + \sum_{\alpha\in\N^n, \abs{\alpha}= d+1}	\frac{\partial^{\alpha} f(z)}{\alpha!} (y-x)^{\alpha}
	\]
	where $z = cx+(1-c)y$ for some $c\in(0,1)$.
\end{theorem}

\vspace{-1em}
\paragraph{Bump Function.}
Consider the bump function $\Psi:\R\to\R$ defined by
$
	\Psi(x) = \begin{cases}
		e^{\frac{1}{x^2-1}}, &\text{ if } \abs{x}<1,\\
		0, &\text{ if } \abs{x}\geq 1.
	\end{cases}
$
It is well known that this function is infinitely differentiable and the derivatives are bounded.
\begin{fact}\label{fact:deri1}
	For all $t\in \N$, $\abs{\Psi^{(t)}(x) } \leq t^{(3+o(1))t}$.
\end{fact}
\noindent Let $\rho$ be the smooth univariate function defined by
$
	\rho(x) = \begin{cases}
		1, &\text{ for }x\geq 1,\\
		e\cdot e^{\frac{1}{(t-1)^2-1}} &\text{ for } 0<x<1,\\
		0, &\text{ for }x\leq 0.
	\end{cases}
$

\noindent It is easy to see $\rho$ is obtained from $\Psi$ via translation, stretch, and concatenation.
We have
\begin{fact}\label{fact:deri2}
	For all $t\in \N$, $\abs{\rho^{(t)}(x)} \leq t^{(3+o(1))t}$.
\end{fact}

\begin{fact}\label{fact:deri}
	Let $r(u,v) \coloneq \rho( \log u-\log v + c )$ for some constant $c$. Then we have that for all $n,m\in \N$,
	$ \abs{ \frac{ \partial^n \partial^m r(u,v) }{ \partial u^n \partial v^m } } \leq \frac{(n+m)^{6(n+m)}}{\abs{u}^n \abs{v}^m} $.
\end{fact}
\noindent We include the proof for the above three facts in \app{fact_bump} for self-containment.

\paragraph{Gaussian Space and the Gaussian Noise Operator}
We denote by $y\sim \NN(0,1)^n$ that $y = (y_1,\dots,y_n)\in\R^n$ is a random vector
whose components are independent standard Gaussian variables (i.e., with mean $0$ and variance $1$).
We say a random vector $Y\in\R^n$ is a $k$-wise independent standard Gaussian vector if every component of $Y$ is a standard Gaussian variable and
$\E[p(Y)] = \E_{y\sim \NN(0,1)^n}[p(y)]$ for all polynomials $p:\R^n\to \R$ with degree at most $k$.
For a function $f:\R^n\to \R$ on Gaussian space and $1\leq p\leq \infty$, the $p$-norm is denoted by
$
	\norm{f}_p = \br{\expect{y\sim \NN(0,1)^n}{\abs{f(y)}^p}}^{1/p}.
$
For $\rho\in[0,1]$, the \emph{Gaussian noise operator} $U_{\rho}$ is the operator on the space of functions $f:\R^n\to\R$ defined by
$
	U_{\rho}f(x) = \expect{y\sim \NN(0,1)^n}{f(\rho x+\sqrt{1-\rho^2}y)}.
$

The \emph{probabilists' Hermite polynomials}~\cite[Section 11]{O14} $\st{H_j}_{j\in\N}$ are defined by
\[
	H_j(y) = \frac{(-1)^{j}}{\varphi(y)}\cdot \frac{\mathrm{d}^j\varphi(y)}{\mathrm{d}\ y^j}
\]
where $\varphi(y)= \frac{1}{\sqrt{2\pi}}e^{-\frac{y^2}{2}}$.
The \emph{univariate Hermite polynomials} $\st{h_j}_{j\in\N}$ are defined by normalization: $h_{j} = \frac{1}{\sqrt{j!}}H_j$.
For a multi-index $\alpha\in\N^n$, the \emph{(multivariate) Hermite polynomial} $h_\alpha:\R^n\to \R$ is
\[
	h_\alpha(y) = \prod_{j=1}^n h_{\alpha_j}(y_j) \enspace.
\]
The degree of $h_\alpha$ is $\abs{\alpha}$.
The Hermite polynomials $\st{h_\alpha}_{\alpha\in\N^n}$ form an orthonormal basis for the functions over Gaussian space: $\expect{y\sim \NN(0,1)^n}{h_\alpha(y)h_\beta(y)} = 1$ iff $\alpha=\beta$, and every degree-$d$ polynomial $f:\R^n\to \R$ can be uniquely expanded as
\[
	f(y) = \sum_{\alpha\in\N^n, \abs{\alpha}\leq d} \widehat{f}(\alpha) h_{\alpha}(y)\enspace.
\]
We can also expand the function $f(x+\sqrt{\lambda} y)$ in the Hermite basis in a manner similar to Taylor expansion.
\begin{lemma}[Lemma 16 in \cite{KM22}]\label{lem:expanding}
Suppose $f(y) = \sum_{\alpha\in\N^n} \widehat{f}(\alpha) h_{\alpha}(y)$, we have
\[
	f(x+\sqrt{\lambda} y) =
	\sum_{\alpha\in\N^n} \frac{\partial^{\alpha}\phi(x)}{\sqrt{\alpha!}} \lambda^{\abs{\alpha}/2} h_{\alpha}(y) \enspace,
\]
where $\phi(x) = U_{\sqrt{1-\lambda}} f\br{\frac{x}{\sqrt{1-\lambda}}}$.
\end{lemma}

The function $U_{\rho}f$ has the following expansion:
\[
	U_\rho f(y) = \sum_{\alpha\in\N^n, \abs{\alpha}\leq d} \rho^{\abs{\alpha}}\widehat{f}(\alpha) h_{\alpha}(y)\enspace.
\]
The definition of $U_\rho$ can be extended to $\rho>1$ by its action on the Hermite polynomials: $U_{\rho}h_{\alpha}(y) = \rho^{\abs{\alpha}}h_{\alpha}(y)$.
We will use the following hypercontractive inequality:
\begin{theorem}\label{thm:hc}
	Let $f:\R^n\to\R$ and $2\leq p\leq \infty$, $\norm{f}_p \leq \norm{U_{\sqrt{p-1}}f}_2$.
\end{theorem}
For more details on analysis over Gaussian space, readers may refer to \cite{O14}. 

\paragraph{Low-Degree Polynomials.}
Low-degree polynomials are extensively studied in the literature.
We list some results used in this paper.
It is well-know that low-degree polynomials have the following anti-concentration property:
\begin{lemma}[Theorem 8 in \cite{CW01}]\label{lem:anti}
	Let $p:\R^n\to\R$ be a polynomial of degree $d$ with $\norm{p}_2=1$. Then 
	\[
		\Pr_{x\sim \NN(0,1)^n} [\abs{p(x)}\leq \epsilon] = O(d\epsilon^{1/d}) \enspace .
	\]
\end{lemma}

Suppose $p$ is a low-degree polynomial, the following gives an estimation on the deviation of $p(x)$ caused by a small perturbation.
\begin{lemma}[Lemma 22 in \cite{Kan11b}] \label{lem:close}
	Let $p:\R^n\to\R$ be a polynomial of degree $d$ with $\norm{p}_2=1$. Suppose $x\in\R^n$ be a vector with $\norm{x}_\infty\leq B(B>1)$. Let $x'$ be another vector such that $\norm{x-x'}_\infty\leq \delta<1$. Then 
	\[
		\abs{p(x)-p(x')}\leq \delta n^{d/2}O(B)^d \enspace .
	\]
\end{lemma}

The magnitudes of the derivatives of a low-degree polynomial are likely to grow at a moderate rate with high probability. Formally,
\begin{lemma}[Lemma 6 in \cite{KM22}]\label{lem:good_event}
	Let $p:\R^n \to \R$ be an arbitrary polynomial of degree $d$ and $y\sim\NN\br{0,1}^n$,
	the following holds with probability at least $1-\epsilon d^3$:
	\[
		\norm{\grad{t}{p(y)}} \leq O\br{\frac{1}{\epsilon}} \norm{\grad{t-1}{p(y)}}
		\text{ for all } 1\leq t\leq d.
	\]
\end{lemma}

The following lemma gives quantitative bounds on how much
the derivatives $\grad{t}{p(x+\sqrt{\lambda}y)}$ are concentrated around those of $\phi(x) =\expect{y\sim \NN(0,1)^n}{p(x+\sqrt{\lambda} y)}$ when $y\sim \NN(0,1)^n$.

\begin{lemma}[Lemma 23 in \cite{KM22}]\label{lem:concentrate}
	Let $0\leq \lambda<1$ and $p:\R^n \to \R$ be an arbitrary polynomial of degree $d$ and $\phi(x) = U_{\sqrt{1-\lambda}} p\!\br{\frac{x}{\sqrt{1-\lambda}}}=\expect{y\sim \NN(0,1)^n}{p(x+\sqrt{\lambda} y)}$. For $0\leq t\leq d$ and $y\sim\NN\br{0,1}^n$,
	\[
	\br{\expect{y\sim\NN\br{0,1}^n}{\norm{ \grad{t}{p(x+\sqrt{\lambda}y)} - \grad{t}{\phi(x)}}^R}}^{\frac{1}{R}} \leq
	\sqrt{ \sum_{j=t+1}^d (\lambda dR)^{j-t} \norm{ \grad{j}{\phi(x)}}^2 } \enspace.
	\]
\end{lemma}

%% file: fool.tex
In this section, we show that a random Gaussian vector matching certain moments fools \emph{any} function of low-degree polynomial threshold functions. Formally, we prove
\begin{theorem}\label{thm:cont}
	Fix a small constant $0<\epsilon<1$ and let $R\in\N$ be an integer. Let $p_1,\dots,p_k:\R^n \to \R$ be arbitrary polynomials of degree $d$ and $f:\bit{k}\to\st{0,1}$ be an arbitrary Boolean function.
	Define function
	\[
		F(x)\coloneqq f\!\br{\sign{p_1\!(x)},\dots,\sign{p_k\!(x)}}
	\]
	Let $Y = \frac{1}{\sqrt{L}} \sum_{i=1}^{L} Y_i $ where $Y_i$ is
	 a $2dR$-wise independent standard Gaussian vector of length $n$ and $L=\Omega\br{\frac{k^{2}d^{3}R^{15}}{\epsilon^2}}$. Then, we have
	\[
	\abs{
		\expect{Y}{
			F(Y)	
		}
		-
		\expect{y\sim\NN\br{0,1}^n}{
			F(y)
		}
	} = O(\epsilon kd^3) + kdL\cdot2^{-\Omega(R)} \enspace .
	\]
	\end{theorem}

The key idea in the proof of \thm{cont} is 
to analyze the derivatives of the disturbed function
$\phi_i(x) = \E_{y\sim\NN\br{0,1}^n}[p_i(x+\sqrt{\lambda}y)]$.
We will see that
once the derivatives of $\phi_i$ are well-controlled by its preceding order derivative at $x$,
$\grad{t}{p_i(x+\sqrt{\lambda}y)}$ is concentrated around
$\grad{t}{\phi_i(x)}$
for a random $y$,
and $p_i(x+\sqrt{\lambda}y)$ and $\phi_i(x)$
share the same sign with high probability.
Starting from this point, we use the mollifier introduced in~\cite{KM22}
\begin{align}
	G(x)\coloneqq \prod_{i=1}^{k} \prod_{t=0}^{d-1} \rho\!\br{\log\!\br{\frac{\norm{\grad{t}{p_i(x)}}^2}{16\epsilon^2 \norm{\grad{t+1}{p_i(x)}}^2 }}} \label{eq:molli}
\end{align}
to judge whether the derivatives are all well-controlled for all $k$ polynomials.
$G(x)=0$ as long as a certain order of derivative that is not controlled by its preceding order derivative.
Our proof consists of following steps:
\begin{itemize}
	\item \textbf{Approximation using the mollifier $G$:}
	We first establish that
	\[
	\abs{
		\expect{Y}{
			F(Y)	
		}
		-
		\expect{y}{
			F(y)
		}
	} \approx 
	\abs{
		\expect{Y}{
			F(Y)G(Y)	
		}
		-
		\expect{y}{
			F(y)G(y)
		}
	}\enspace.
	\]
	This approximation enables us to focus primarily on the analysis of $F(y)G(y)$ in the subsequent steps.
	\item \textbf{Hybrid argument:} Let $\lambda = L^{-1}$, $y = \sqrt{\lambda}\sum_{i=1}^L y_i$ where $y_i\sim \NN(0,1)^n$ and $Z^i =\sqrt{\lambda}(y_1 + \cdots +y_{i-1} +Y_{i+1} +\cdots Y_L)$. We will show 
	\begin{align}
		\expect{ }{
			F(Z^i+\sqrt{\lambda}Y_i)G(Z^i+\sqrt{\lambda}Y_i)	
		}
		\approx
		\expect{ }{
			F(Z^i+\sqrt{\lambda}y_i)G(Z^i+\sqrt{\lambda}y_i)
		} \enspace.\label{eq:singlestep}
	\end{align}
	Therefore by the triangle inequality, we have
	\begin{align*}
		\expect{Y}{
			F(Y)G(Y)	
		} &= \expect{ }{
			F(Z^1+\sqrt{\lambda}Y_1)G(Z^1+\sqrt{\lambda}Y_1)	
		}\\
		&\approx
		\expect{ }{
			F(Z^1+\sqrt{\lambda}y_1)G(Z^1+\sqrt{\lambda}y_1)
		}
		= \expect{ }{
			F(Z^2+\sqrt{\lambda}Y_2)G(Z^2+\sqrt{\lambda}Y_2)	
		}\\
		&\approx \cdots \approx
		\expect{ }{
			F(Z^L+\sqrt{\lambda}y_L)G(Z^L+\sqrt{\lambda}y_L)
		}= \expect{ }{
			F(y)G(y)
		} \enspace.
	\end{align*}
	To prove \eq{singlestep}, we show for any fixed $x$,
	\[
		\expect{}{
			F(x+\sqrt{\lambda}Y_i)G(x+\sqrt{\lambda}Y_i)
		}
		\approx
		\expect{}{
			F(x+\sqrt{\lambda}y_i)G(x+\sqrt{\lambda}y_i)
		} \enspace .
	\]
	This is done by a case analysis:
	\begin{itemize}
	\item The derivatives of all $k$ polynomials $\phi_j(x)$ are well-controlled at point $x$. In this case, all $p_j(x+\sqrt{\lambda} Y_i)$ and $p_j(x+\sqrt{\lambda} y_i)$ share the same sign with high probability. Thus, it is highly likely that $F(x+\sqrt{\lambda}Y_i)$ and $F(x+\sqrt{\lambda}y_i)$ are nearly the same constant. It suffices to show $Y_i$ fools the mollifier function $G(x+\sqrt{\lambda}y_i)$.
	\item At least one derivative is not controlled. In this case, we will show that $G(x+\sqrt{\lambda}Y_i)$ and $G(x+\sqrt{\lambda}y_i)$ are $0$ with high probability. This implies that $F(x+\sqrt{\lambda}Y_i)G(x+\sqrt{\lambda}Y_i) = F(x+\sqrt{\lambda}y_i)G(x+\sqrt{\lambda}y_i) = 0$ with overwhelming probability.
	\end{itemize}
	\end{itemize}
	
	In the subsequent sections, \sect{molli} first demonstrates that $Y_i$ is able to fool the mollifier function $G$ when $x$ is a well-behaved point.
	\sect{single} shows the closeness of a single step in the hybrid argument.
	Lastly, we prove \thm{cont} using approximation and the hybrid argument in \sect{main}.

\subsection{Fooling the Mollifier $G$} \label{sec:molli}
We begin with proving that a $2dR$-wise independent standard Gaussian vector $Y$
fools the mollifier function $G(x+\sqrt{\lambda}y)$.
To achieve this, we utilize the Taylor expansion to expand
the mollifier function
$G(x+\sqrt{\lambda}y)$ up to a specified order.
As a result, $G(x+\sqrt{\lambda}y)$ is decomposed into two parts:
a degree-$d(R-1)$ polynomial $l(y)$ and a remainder term $\Delta(y)$.
We mainly show that $\E[\Delta]$ is negligible under both
pseudorandom distribution and true Gaussian distribution.
This leads us to the conclusion that
$\E[G(x+\sqrt{\lambda}y)]\approx \E[l(y)]$ and
$\E[G(x+\sqrt{\lambda}Y)]\approx \E[l(Y)]$.
Furthermore, since $l(y)$ has degree at most $dR$,
it follows that $\E[l(y)] = \E[l(Y)]$.
Thus, we conclude that $\E[G(x+\sqrt{\lambda}y)]\approx\E[G(x+\sqrt{\lambda}Y)]$.

\begin{lemma}\label{lem:molli}
		Fix a small constant $0<\epsilon<1$ and let $R\in\N$ be an integer.
		Let $p_1,\dots,p_k:\R^n \to \R$ be arbitrary polynomials of degree $d$.
	Define $\phi_i(x) \coloneq U_{\sqrt{1-\lambda}} p_i\!\br{\frac{x}{\sqrt{1-\lambda}}}=\E_{y\sim\NN\br{0,1}^n}[p_i(x+\sqrt{\lambda}y)]$ for all $p_i$.
	Suppose that a fix point $x\in\R^n$ satisfies $ \norm{\grad{t+1}{\phi_i(x)}} \leq \frac{1}{\epsilon}\norm{\grad{t}{\phi_i(x)}}$ for any $1\leq i\leq k$ and $0\leq t\leq d-1$.
	Let $Y$ be a $2dR$-wise independent standard Gaussian vector of length $n$.
	For $\lambda = O(k^{-2}d^{-3}R^{-15}\epsilon^2)$, we have 
	\[
	\abs{
		\expect{Y}{
			G(x+\sqrt{\lambda}Y)
		}
		-
		\expect{y\sim\NN\br{0,1}^n}{
			G(x+\sqrt{\lambda}y)
		}
	} = kd\cdot 2^{-\Omega(R)} \enspace ,
	\]
	where $G$ is defined in \eq{molli}.
	\end{lemma}
	
\begin{proof}
	Let $\sigma(z) \coloneq \rho(z - \log{16\epsilon^2})$ and by the definition of function $G\br{\cdot}$ defined in \eqref{eq:molli},
	\[G(z) = \prod_{i=1}^{k} \prod_{t=0}^{d-1} \sigma\!\br{ \log \norm{\grad{t}{p_i(z)}}^2 - \log \norm{\grad{t+1}{p_i(z)}}^2 }.\]
	Define variables $\st{s_i^t}_{1\leq i\leq k, 0\leq t\leq d-1}$ and $\st{r_i^t}_{1\leq i\leq k, 0\leq t\leq d-1}$ by letting $s_i^t = \norm{\grad{t}{p_i(x+\sqrt{\lambda}y)}}^2$ and $r_i^t = \norm{\grad{t+1}{p_i(x+\sqrt{\lambda}y)}}^2$ as functions of $y$.
	Apparently, we have
	\[
		G(x+\sqrt{\lambda}y) = g(s,r) \coloneq 
		\prod_{i=1}^{k} \prod_{t=0}^{d-1} \sigma\!\br{ \log s_i^t - \log r_i^t } \enspace .
	\]
	Therefore, it is equivalent to prove
	\[
		\abs{
		\expect{y\sim\NN\br{0,1}^n}{
			g(s,r)
		}
		-
		\expect{Y}{
			g(s,r)
		}
	} = kd\cdot 2^{-\Omega(R)} \enspace .
	\]
	To this end, we expand $g(s,r)$ into $R$-th order using the Taylor expansion at some point $(a,b)$:
	$
		g(s,r) = l(s,r) + \Delta,
	$
	where $l(s,r)$ is a polynomial of $y$ of degree at most $dR$ and $\Delta$ is the remainder.
	Since $Y$ is $2dR$-wise independent, we know
	$
		\expect{y\sim\NN\br{0,1}^n}{
			l(s,r)
		}
		=
		\expect{Y}{
			l(s,r)
		}
	$.
	Therefore, it suffices to show $\expect{y\sim\NN\br{0,1}^n}{\abs{\Delta}}$ and $\expect{Y}{\abs{\Delta}}$ are bounded by $kd\cdot 2^{-\Omega(R)}$.
	
	More specifically, we choose to expand $g(s,r)$ at points $a_i^t = \norm{\grad{t}{\phi_i(x)}}^2$ and $b_i^t = \norm{\grad{t+1}{\phi_i(x)}}^2$ via \thm{taylor}:
	\[
		g(s,r) = l(s,r) + \Delta \enspace,
	\]
	where
	\[
		l(s,r) = \sum_{\substack{ \br{\alpha^{t}_i}_{i\in[k], 0\leq t\leq d-1} \in\mathbb{N}^{kd}\\
		\br{\beta^{t}_i}_{i\in[k], 0\leq t\leq d-1}\in\mathbb{N}^{kd}
		\\ \abs{\alpha} +\abs{\beta} <R }}	
		\frac{\partial^\alpha_s\partial^\beta_r g(a,b)}{\alpha!\beta!}(s-a)^\alpha(r-b)^\beta
	\]
	and
	\[
		\Delta = \sum_{\substack{ \br{\alpha^{t}_i}_{i\in[k], 0\leq t\leq d-1} \in\mathbb{N}^{kd}\\
		\br{\beta^{t}_i}_{i\in[k], 0\leq t\leq d-1}\in\mathbb{N}^{kd}
		\\ \abs{\alpha} +\abs{\beta} =R }}	
		\frac{ \partial^\alpha_s\partial^\beta_r g(s^*,r^*) }{\alpha!\beta!}(s-a)^\alpha(r-b)^\beta \enspace .
	\]
	for some $(s^*,r^*)$ on the line segment joining $(s,r)$ and $(a,b)$.
	It is not hard to see that
	$l(s,r)$ is a polynomial of $y$ of degree at most $d(R-1)$.
	For the remainder term $\Delta$, we will prove the following bound:
\[
	\expect{y\sim\NN\br{0,1}^n}{\abs{\Delta}} \leq kd\cdot 2^{-\Omega(R)} \enspace.
\]

	We now give bounds on $\expect{y\sim\NN\br{0,1}^n}{\abs{\Delta}}$.
	The same argument applies to $Y$ as well.
	Fix a small constant $0<\delta < \frac{1}{kdR^7} $. Let $\EE_i^t$ be the event that
	\[
		\abs{ \norm{\grad{t}{p_i(x+\sqrt{\lambda}y)}} -  \norm{\grad{t}{\phi_i(x)}}} \leq \delta \norm{\grad{t}{\phi_i(x)}} \enspace.
	\]
	Now let $\EE = \wedge_{1\leq i\leq k, 0\leq t\leq d} \EE_{i}^t$.
	Note that
	\[ \Delta = \Delta \cdot \id_\EE + \Delta \cdot \id_{\widebar{\EE}} = \Delta \cdot \id_\EE + g(s,r)\cdot \id_{\widebar{\EE}} - l(s,r)\cdot \id_{\widebar{\EE}} \enspace . \]
	Here, $\id_\A = 1$ when event $\A$ occurs and $\id_\A = 0$ otherwise.
	Therefore, by the triangle inequality, we have
	\begin{align*}
		\expect{y}{\abs{\Delta}}&\leq
		\expect{y}{\Delta \cdot \id_\EE} + 
		\expect{y}{g(s,r)\cdot \id_{\widebar{\EE}}} + 
		\expect{y}{\abs{l(s,r)}\cdot \id_{\widebar{\EE}}} \nonumber\\ 
		&\leq
		\expect{y}{\Delta \cdot \id_\EE} + 
		\expect{y}{\id_{\widebar{\EE}}} + 
		\expect{y}{\abs{l(s,r)}\cdot \id_{\widebar{\EE}}} \tag*{\llap{($g(s,r)\in [0,1]$)}}\nonumber\\
		&\leq
		\underbrace{\expect{y}{\Delta \cdot \id_\EE}}_{(\text{Term }1)} +
		\underbrace{\expect{y}{\id_{\widebar{\EE}}}}_{(\text{Term }2)} +
		\underbrace{\sqrt{\expect{y}{l^2(s,r)}}}_{(\text{Term }3)}\cdot
		\sqrt{ \expect{y}{\id_{\widebar{\EE}}}} \enspace . \tag*{\llap{(Cauchy–Schwarz)}}
	\end{align*}
	We are next to bound Term $1\sim 3$.
	
	\paragraph{Bounding Term 1.}
	If event $\EE$ occurs,
	we have
	\begin{align} \label{eq:bound_E}
		\abs{\Delta} &\leq \sum_{\substack{ \br{\alpha^{t}_i}\in\mathbb{N}^{kd}, \br{\beta^{t}_i}\in\mathbb{N}^{kd}
		\\ \abs{\alpha} +\abs{\beta} =R }}
		\frac{ \abs{\partial^\alpha_s\partial^\beta_r g(s^*,r^*)} }{\alpha!\beta!}
		\prod_{i=1}^{k} \prod_{t=0}^{d-1} \abs{s_i^t-a_i^t}^{\alpha_i^t} \abs{r_i^t-b_i^t}^{\beta_i^t} \nonumber\\
		&\leq \sum_{\substack{ \br{\alpha^{t}_i}\in\mathbb{N}^{kd}, \br{\beta^{t}_i}\in\mathbb{N}^{kd}
		\\ \abs{\alpha} +\abs{\beta} =R }} { R^{6R} }\cdot 
		\prod_{i=1}^{k} \prod_{t=0}^{d-1} \frac{\abs{s_i^t-a_i^t}^{\alpha_i^t}}{\abs{s^{*t}_i}^{\alpha_i^t}} \frac{\abs{r_i^t-b_i^t}^{\beta_i^t}}{\abs{r^{*t}_i}^{\beta_i^t}} \nonumber\\
		&\leq \sum_{\substack{ \br{\alpha^{t}_i}\in\mathbb{N}^{kd}, \br{\beta^{t}_i}\in\mathbb{N}^{kd}
		\\ \abs{\alpha} +\abs{\beta} =R }} { R^{6R} }\cdot \br{\frac{2\delta + \delta^2}{(1-\delta)^2}}^{R}
		= \binom{R+2kd-1}{R}\cdot { R^{6R} }\cdot \br{4\delta}^{R} \leq 2^{-R}\enspace,\nonumber 
	\end{align}
where the second inequality is from \fct{deri}, 
and the third one is true by the following facts:
\begin{itemize}
	\item $(1-\delta)^2 a^t_i \leq s_i^t\leq (1+\delta)^2 a^t_i$, $(1-\delta)^2 b^t_i \leq r_i^t \leq (1+\delta)^2 b^t_i$,
	\item ${s^{*t}_i}\geq \min\st{s^{t}_i, a^{t}_i}\geq (1-\delta)^2a^{t}_i$ and ${r^{*t}_i}\geq \min\st{r^{t}_i, b^{t}_i}\geq (1-\delta)^2b^{t}_i$, since $(s^{*t}_i, r^{*t}_i)$ lies between $(s^{t}_i,r^{t}_i)$ and $(a^{t}_i,b^{t}_i)$.
\end{itemize}
	This gives us
	\[
		\expect{y}{\Delta \cdot \id_\EE}\leq 2^{-R}\enspace.
	\]
	
	\paragraph{Bounding Term 2.}
	Since
	\[
		\norm{\grad{t}{p_i(x+\sqrt{\lambda}y)}} -  \norm{\grad{t}{\phi_i(x)}} \leq \norm{\grad{t}{p_i(x+\sqrt{\lambda}y)} - \grad{t}{\phi_i(x)}}
	\]
	and
	\[
		\norm{\grad{t}{\phi_i(x)}} - \norm{\grad{t}{p_i(x+\sqrt{\lambda}y)}} \leq \norm{\grad{t}{p_i(x+\sqrt{\lambda}y)} - \grad{t}{\phi_i(x)}}\enspace,
	\]
	we have
	\begin{equation}\label{eq:eq1}
		\abs{ \norm{\grad{t}{p_i(x+\sqrt{\lambda}y)}} -  \norm{\grad{t}{\phi_i(x)}}} \leq \norm{\grad{t}{p_i(x+\sqrt{\lambda}y)} - \grad{t}{\phi_i(x)}}
	\end{equation}
	This gives us
	\begin{align*}
		\Pr_{y} \Br{\EE_i^t}&\geq \Pr_y\Br{ \norm{\grad{t}{p_i(x+\sqrt{\lambda}y)} - \grad{t}{\phi_i(x)}} \leq \delta \norm{\grad{t}{\phi_i(x)}}} \\
		&\geq 1 - \frac{\expect{y}{\norm{ \grad{t }{p_{i }(x+\sqrt{\lambda}y)} - \grad{t }{\phi_{i  }(x)}}^R}}{ \delta^{R}\norm{ \grad{t }{\phi_{i }(x)}}^R} \\
		&\geq 1 - \frac{\br{ \sum_{j=t+1}^d (\lambda d R)^{j-t} \norm{ \grad{j}{\phi(x)}}^2 }^{\frac{R}{2}} }{\delta^{R}\norm{ \grad{t }{\phi_{i }(x)}}^R} \\
		&\geq 1 - \br{\frac{1}{\delta}}^R \br{\sum_{j=1}^{d-t} \br{\frac{\lambda d R}{\epsilon^2}}^{j}}^{\frac{R}{2}} \geq 1-\br{\sqrt{\frac{2\lambda d R}{\delta^2\epsilon^2}}}^R
	\end{align*}
	where the second inequality is from Markov's inequality, the third one is from \lem{concentrate} and the fourth one holds since $x\in\R^n$ satisfies $ \norm{\grad{t+1}{\phi_i(x)}} \leq \frac{1}{\epsilon}\norm{\grad{t}{\phi_i(x)}}$ for any $1\leq i\leq k$ and $0\leq t\leq d-1$.
	For $\lambda = O(k^{-2}d^{-3}R^{-15}\epsilon^2)$, we have
	$
		\Pr_{y} \Br{\EE_i^t} \geq 1-2^{-R}.
	$
	Consequently, $\Pr_{y} \Br{\EE} \geq 1-kd2^{-R}.$
	Therefore,
	\[
		\expect{y}{\id_{\widebar{\EE}}} \leq kd2^{-R} \enspace.
	\]
	
	\paragraph{Bounding Term 3.}
	We next to upper bound $\expect{y}{l^2(s,r)}$. Note that
	\begin{align*}
		\expect{y}{l^2(s,r)}
		&\leq \sum_{\substack{ \abs{\alpha} +\abs{\beta} <R
					\\ \abs{\alpha'} +\abs{\beta'} <R }}
		\frac{\abs{\partial^\alpha_s\partial^\beta_r g(a,b)}}{\alpha!\beta!}\frac{\abs{\partial^{\alpha'}_s\partial^{\beta'}_r g(a,b)}}{\alpha'!\beta'!}
			\expect{y}{\abs{(s-a)^\alpha(r-b)^\beta(s-a)^{\alpha'}(r-b)^{\beta'}}} \\
			\
		&\leq
		\sum_{\substack{q_1<R\\q_2<R}}
		\sum_{\substack{ \abs{\alpha} +\abs{\beta} =q_1
					\\ \abs{\alpha'} +\abs{\beta'} =q_2 }}
					R^{6(q_1+q_2)}\cdot
		\underbrace{
			\expect{y}{ \prod_{i=1}^{k} \prod_{t=0}^{d-1}
			\frac{\abs{s_i^t-a_i^t}^{\alpha_i^t}}{\abs{a^{t}_i}^{\alpha_i^t}}
			\frac{\abs{r_i^t-b_i^t}^{\beta_i^t}}{\abs{b^{t}_i}^{\beta_i^t}} 
			\frac{\abs{s_i^t-a_i^t}^{{\alpha'}_i^t}}{\abs{a^{t}_i}^{{\alpha'}_i^t}}
			\frac{\abs{r_i^t-b_i^t}^{{\beta'}_i^t}}{\abs{b^{t}_i}^{{\beta'}_i^t}} } }_
			{(\star)}
	\end{align*}
	By generalized H\"older's inequality, for $q_1+q_2\neq 0$
	\begin{align*}
		&(\star) \\
		\leq &\prod_{i,t}
			\frac{\br{\expect{y}{\abs{s_i^t-a_i^t}^{q_1+q_2}}}^{\frac{\alpha_i^t}{q_1+q_2}}}{\abs{a^{t}_i}^{\alpha_i^t}}
			\frac{\br{\expect{y}{\abs{r_i^t-b_i^t}^{q_1+q_2}}}^{\frac{\beta_i^t}{q_1+q_2}}}{\abs{b^{t}_i}^{\beta_i^t}} 
			\frac{\br{\expect{y}{\abs{s_i^t-a_i^t}^{q_1+q_2}}}^{\frac{{\alpha'}_i^t}{q_1+q_2}}}{\abs{a^{t}_i}^{{\alpha'}_i^t}}
			\frac{\br{\expect{y}{\abs{r_i^t-b_i^t}^{q_1+q_2}}}^{\frac{{\beta'}_i^t}{q_1+q_2}}}{\abs{b^{t}_i}^{{\beta'}_i^t}}\enspace  .
	\end{align*}
	Note that for $0<q\leq 2R$,
	\begin{align*}
		&\expect{y}{\br{\frac{\abs{\norm{\grad{t}{p_i(x+\sqrt{\lambda}y)}}^2 - \norm{\grad{t}{\phi_i(x)}}^2}}{\norm{\grad{t}{\phi_i(x)}}^2}}^q}\\
		\leq &\expect{y}{\br{
		\frac{
		2\abs{\norm{\grad{t}{p_i(x+\sqrt{\lambda}y)}} -  \norm{\grad{t}{\phi_i(x)}}}\cdot \norm{\grad{t}{\phi_i(x)}} + \abs{\norm{\grad{t}{p_i(x+\sqrt{\lambda}y)}} -  \norm{\grad{t}{\phi_i(x)}}}^2
		}
		{\norm{\grad{t}{\phi_i(x)}}^2}}^q}\\
		\leq &\expect{y}{\br{
		\frac{
		2\norm{\grad{t}{p_i(x+\sqrt{\lambda}y)}-\grad{t}{\phi_i(x)}}
		}
		{\norm{\grad{t}{\phi_i(x)}}}
		+
		\frac{
		\norm{\grad{t}{p_i(x+\sqrt{\lambda}y)}-\grad{t}{\phi_i(x)}}^2
		}
		{\norm{\grad{t}{\phi_i(x)}}^2}
		}^q}\\
		= &\sum_{j=0}^q 2^j \expect{y}{\frac{
		\norm{\grad{t}{p_i(x+\sqrt{\lambda}y)}-\grad{t}{\phi_i(x)}}^{2q-j}
		}
		{\norm{\grad{t}{\phi_i(x)}}^{2q-j}}}
		\leq \sum_{j=0}^q 2^j \cdot \br{\sqrt{\frac{4\lambda d R}{\epsilon^2}}}^{2q-j}
		= 4^q\cdot\sum_{j=0}^q \br{\sqrt{\frac{\lambda d R}{\epsilon^2}}}^{2q-j}\\
		\leq & 2\cdot 4^q\cdot \br{\sqrt{\frac{\lambda d R}{\epsilon^2}}}^{q}
		\leq \br{\sqrt{\frac{17\lambda d R}{\epsilon^2}}}^{q}
	\end{align*}
	where the first inequality is from $\abs{a^2-b^2} = \abs{a-b}\abs{a+b}\leq  \abs{a-b}(\abs{a} + \abs{b}) \leq \abs{a-b}(2\abs{a}+\abs{a-b})$, the second one is from Eq.\eq{eq1} and the third inequality is from \lem{concentrate}.
	Therefore we have
	\begin{align*}
		(\star) \leq \prod_{i,t} \br{\sqrt{\frac{17\lambda d R}{\epsilon^2}}}^{ \alpha_i^t+\beta_i^t +{\alpha'}_i^t+{\beta'}_i^t } = \br{\sqrt{\frac{17\lambda d R}{\epsilon^2}}}^{ q_1+q_2} \enspace .
	\end{align*}
	Consequently,
	\begin{align*}
		\expect{y}{l^2(s,r)} &\leq \sum_{\substack{q_1<R\\q_2<R}}
		\sum_{\substack{ \abs{\alpha} +\abs{\beta} =q_1
					\\ \abs{\alpha'} +\abs{\beta'} =q_2 }}
					R^{6(q_1+q_2)}\cdot\br{\sqrt{\frac{17\lambda d R}{\epsilon^2}}}^{q_1+q_2} \\
		&\leq \sum_{\substack{q_1<R\\q_2<R}}
				(R+2kd-1)^{q_1+q_2}\cdot R^{6(q_1+q_2)}\cdot\br{\sqrt{\frac{17\lambda d R}{\epsilon^2}}}^{q_1+q_2}		\\
		&= \sum_{q=0}^{2R-2} (q+1)\cdot\br{(R+2kd-1)R^6\sqrt{\frac{17\lambda d R}{\epsilon^2}}}^q\enspace .
	\end{align*}
	For $\lambda = O(k^{-2}d^{-3}R^{-15}\epsilon^2)$ is sufficiently small, we have
	$
		\expect{y}{l^2(s,r)} \leq \sum_{q=0}^{2R-2} (q+1)\cdot 4^{-q} < 2R.
	$
	
	Thus, putting everything together, we have
	\begin{align*}
		\expect{y}{\abs{\Delta}}
		&\leq \expect{y}{\Delta \cdot \id_\EE} + \expect{y}{\id_{\widebar{\EE}}} + \sqrt{\expect{y}{l^2(s,r)}} \sqrt{ \expect{y}{\id_{\widebar{\EE}}}} \\
		&\leq 2^{-R} + kd2^{-R} + \sqrt{2Rkd2^{-R}} \\
		&\leq kd\cdot 2^{-\Omega(R)}
		\enspace .
	\end{align*}
We now regard $l(s,r)$ as a function of $y$, and from the above inequality we have
	\begin{equation*}
		\abs{\expect{y}{G(x+\sqrt{\lambda}y)} - \expect{y}{l(s,r)}}\leq 
		\expect{y}{\abs{\Delta}}\leq kd\cdot 2^{-\Omega(R)}\enspace . 
	\end{equation*}
	Similarly, the same argument applying on $2dR$-wise independent Gaussian vector $Y$ gives us
	\begin{equation*}
		\abs{\expect{Y}{G(x+\sqrt{\lambda}Y)} - \expect{Y}{l(s,r)}}\leq kd\cdot 2^{-\Omega(R)}\enspace . 
	\end{equation*}
	The lemma then follows from the fact that $\expect{y}{l(s,r)} = \expect{Y}{l(s,r)}$, since $l(s,r)$ is a polynomial of $y$ of degree at most $d(R-1)$.
\end{proof}	

\subsection{A Single Step in the Hybrids} \label{sec:single}
In this section,
we analyze one single step in the entire hybrid argument.
We will show that for any $x$, we have that
$\E_Y[F(x+\sqrt{\lambda}Y)G(x+\sqrt{\lambda}Y)]\approx
\E_y[F(x+\sqrt{\lambda}y)G(x+\sqrt{\lambda}y)]$ for $2dR$-wise independent Gaussian $Y$ and true Gaussian $y$.

Let $\phi_i(x) = U_{\sqrt{1-\lambda}} p_i\!\br{\frac{x}{\sqrt{1-\lambda}}} = \E_y[p_i(x+\sqrt{\lambda}y)]$.
The proof proceeds through a case analysis
based on the behavior of $\phi_i$ at the fixed point $x$.
Specifically, we define $x$ as well-behaved if
$\norm{\grad{t+1}{\phi_i(x)}} \leq \frac{1}{\epsilon}\norm{\grad{t}{\phi_i(x)}}$ for all $t\in[d]$ and $i\in[k]$.
In other words, for each function $\phi_i$,
its $t$-th order derivatives are controlled by its
$(t-1)$-th order derivatives.
\begin{itemize}
	\item In the scenario where $x$ is not well-behaved,
	we can identify an $i_0$ and a $t_0$ such that
	with at least probability $1-2^{-R+1}$,
	\[
	\norm{ \grad{t_0+1}{p_{i_0}(x+\sqrt{\lambda}y)}}> \frac{1}{4\epsilon} \norm{ \grad{t_0}{p_{i_0}(x+\sqrt{\lambda}y)}} \enspace .
	\]
	Thus, it is highly probable that the mollifier function
	$G(x+\sqrt{\lambda}y)=0$.
	So, the expectation of $F(x+\sqrt{\lambda}y)G(x+\sqrt{\lambda}y)$ is no more that $2^{-R+1}$.
	The same argument works for $Y$ as well.
	\item For the case that $x$ is well-behaved,
	we will show that for all $p_i$,
	${ p_i(x+\sqrt{\lambda}y) }$ and ${ p_i(x+\sqrt{\lambda}Y) }$ are nearly the same constant. This implies
	$F(x+\sqrt{\lambda}y)$ and $F(x+\sqrt{\lambda}Y)$ are equal in most situations.
	Then it suffices to show $Y$ fools the mollifier,
	as discussed in the previous section.
\end{itemize}

\begin{lemma}\label{lem:hybrid}
	Fix a small constant $0<\epsilon<1$ and let $R\in\N$ be an integer. Let $p_1,\dots,p_k:\R^n \to \R$ be arbitrary polynomials of degree $d$ and $f:\bit{k}\to\st{0,1}$ be an arbitrary Boolean function.
	Define function
	\[
		F(x)\coloneqq f\!\br{\sign{p_1\!(x)},\dots,\sign{p_k\!(x)}}\enspace.
	\]
	Let $Y$ be a $2dR$-wise independent standard Gaussian vector of length $n$.
	For any $x\in\R^n$ and $\lambda = O(k^{-2}d^{-3}R^{-15}\epsilon^2)$
	\[
	\abs{
		\expect{Y}{
			F(x+\sqrt{\lambda}Y)G(x+\sqrt{\lambda}Y)
		}
		-
		\expect{y\sim\NN\br{0,1}^n}{
			F(x+\sqrt{\lambda}y)G(x+\sqrt{\lambda}y)
		}
	} = kd2^{-\Omega(R)} \enspace ,
	\]
	where $G$ is defined in \eq{molli}.
\end{lemma}

\begin{proof}
	Let $\phi_i(x) = U_{\sqrt{1-\lambda}} p_i\!\br{\frac{x}{\sqrt{1-\lambda}}}$.
	Define $x$ is good if for any $1\leq i\leq k$ and $0\leq t\leq d-1$, $ \norm{\grad{t+1}{\phi_i(x)}} \leq \frac{1}{\epsilon}\norm{\grad{t}{\phi_i(x)}} $.
	We prove this lemma by considering $x$ is good or not.
	
	We first consider that $x$ is not good. 
	In this case, we will show $G(x+\sqrt{\lambda}y) = 0$ holds
	with high probability.
	Consequently, $F(x+\sqrt{\lambda}y)G(x+\sqrt{\lambda}y)$ is zero with high probability.
	To this end, it suffices to find an $i_0$ and a $t_0$ such that
	\[
	\norm{ \grad{t_0+1}{p_{i_0}(x+\sqrt{\lambda}y)}}> \frac{1}{4\epsilon} \norm{ \grad{t_0}{p_{i_0}(x+\sqrt{\lambda}y)}} \enspace .
	\]
	
	We choose an arbitrary $i_0$ satisfying that there exists $0\leq t\leq d-1$ such that $ \norm{\grad{t+1}{\phi_{i_0}(x)}} > \frac{1}{\epsilon}\norm{\grad{t}{\phi_{i_0}(x)}} $.
	Since $x$ is not good, we know such $i_0$ exists. 
	And let $t_0$ be the largest $t$ such that $ \norm{\grad{t+1}{\phi_{i_0}(x)}} > \frac{1}{\epsilon}\norm{\grad{t}{\phi_{i_0}(x)}} $ holds. It is not hard to check
	\begin{itemize}
		\item $ \norm{\grad{t_0}{\phi_{i_0}(x)}} < {\epsilon}\norm{\grad{t_0+1}{\phi_{i_0}(x)}} $,
		\item $ \norm{\grad{t+1}{\phi_{i_0}(x)}} \leq \frac{1}{\epsilon}\norm{\grad{t}{\phi_{i_0}(x)}} $ for $t\geq t_0+1$.
	\end{itemize}
	We are next to prove the following inequalities hold with high probability,
	\begin{itemize}
		\item[(a)] $
	\norm{ \grad{t_0}{p_{i_0}(x+\sqrt{\lambda}y)}} < 2\epsilon\norm{ \grad{t_0+1}{\phi_{i_0}(x)}}$
		\item[(b)] $
	\norm{ \grad{t_0+1}{p_{i_0}(x+\sqrt{\lambda}y)}}> \frac{1}{2}\norm{ \grad{t_0+1}{\phi_{i_0}(x)}}$
	\end{itemize}
	It is easy to see $(a)$ and $(b)$ give us $\norm{ \grad{t_0+1}{p_{i_0}(x+\sqrt{\lambda}y)}}> \frac{1}{4\epsilon} \norm{ \grad{t_0}{p_{i_0}(x+\sqrt{\lambda}y)}}$.
	\vspace{-1em}
	\paragraph{(a)}
	By Markov's inequality, we have
	\begin{align*}
		&\Pr_y\Br{\norm{ \grad{t_0}{p_{i_0}(x+\sqrt{\lambda}y)} - \grad{t_0}{\phi_{i_0}(x)}} \geq \epsilon\norm{ \grad{t_0+1}{\phi_{i_0}(x)}} }
		\leq \frac{\expect{y}{\norm{ \grad{t_0}{p_{i_0}(x+\sqrt{\lambda}y)} - \grad{t_0}{\phi_{i_0}(x)}}^R}}{\epsilon^R\norm{ \grad{t_0+1}{\phi_{i_0}(x)}}^R} \\
		\leq &\frac{\br{ \sum_{t=t_0+1}^d (\lambda d R)^{t-t_0} \norm{ \grad{t}{\phi(x)}}^2 }^{\frac{R}{2}} }{\epsilon^R\norm{ \grad{t_0+1}{\phi_{i_0}(x)}}^R} 
		\leq \frac{\br{ \sum_{t=t_0+1}^d (\lambda d R)^{t-t_0} \br{\frac{1}{\epsilon^2}}^{t-t_0-1} \norm{ \grad{t_0+1}{\phi(x)}}^2 }^{\frac{R}{2}} }{\epsilon^R\norm{ \grad{t_0+1}{\phi_{i_0}(x)}}^R} 
		\leq \br{\sum_{t=1}^{d-t_0} \br{\frac{\lambda d R}{\epsilon^2}}^{t}}^{\frac{R}{2}}
	\end{align*}
	Here the second inequality is from \lem{concentrate}.
	The third inequality uses the condition that $ \norm{\grad{t+1}{\phi_{i_0}(x)}} \leq \frac{1}{\epsilon}\norm{\grad{t}{\phi_{i_0}(x)}} $ for $t\geq t_0+1$.
	Since $\lambda\leq \frac{\epsilon^2}{100dR}$, this probability is bounded by $2^{-R}$. Therefore, with probability at least $1-2^{-R}$,
	\[
		\norm{ \grad{t_0}{p_{i_0}(x+\sqrt{\lambda}y)} - \grad{t_0}{\phi_{i_0}(x)}} < \epsilon\norm{ \grad{t_0+1}{\phi_{i_0}(x)}}.
	\]
	Moreover, we know $ \norm{\grad{t_0}{\phi_{i_0}(x)}} < {\epsilon}\norm{\grad{t_0+1}{\phi_{i_0}(x)}} $.
	So, we have with probability at least $1-2^{-R}$,
	\[
	\norm{ \grad{t_0}{p_{i_0}(x+\sqrt{\lambda}y)}}\leq
	\norm{ \grad{t_0}{\phi_{i_0}(x)}} +
	\norm{ \grad{t_0}{p_{i_0}(x+\sqrt{\lambda}y)} - \grad{t_0}{\phi_{i_0}(x)}} < 2\epsilon\norm{ \grad{t_0+1}{\phi_{i_0}(x)}}\enspace .
	\]
	
	\paragraph{(b)}
	Similarly, we have
	\vspace{-1em}
	\begin{align*}
		&\Pr_y\Br{\norm{ \grad{t_0+1}{p_{i_0}(x+\sqrt{\lambda}y)} - \grad{t_0+1}{\phi_{i_0}(x)}} \geq \frac{1}{2}\norm{ \grad{t_0+1}{\phi_{i_0}(x)}} }
		\leq \frac{\expect{y}{\norm{ \grad{t_0+1}{p_{i_0}(x+\sqrt{\lambda}y)} - \grad{t_0+1}{\phi_{i_0}(x)}}^R}}{ 2^{-R}\norm{ \grad{t_0+1}{\phi_{i_0}(x)}}^R} \\
		\leq &\frac{\br{ \sum_{t=t_0+2}^d (\lambda d R)^{t-t_0-1} \norm{ \grad{t}{\phi(x)}}^2 }^{\frac{R}{2}} }{2^{-R}\norm{ \grad{t_0+1}{\phi_{i_0}(x)}}^R} 
		\leq \frac{\br{ \sum_{t=t_0+2}^d \br{\frac{\lambda d R}{\epsilon^2}}^{t-t_0-1}\norm{ \grad{t_0+1}{\phi(x)}}^2 }^{\frac{R}{2}} }{2^{-R}\norm{ \grad{t_0+1}{\phi_{i_0}(x)}}^R} 
		\leq \br{4 \cdot  \sum_{t=1}^{d-t_0-1} \br{\frac{\lambda d R}{\epsilon^2}}^{t}}^{\frac{R}{2}}
	\end{align*}
	Here the second inequality is from \lem{concentrate}.
	The third inequality uses the condition that $ \norm{\grad{t+1}{\phi_{i_0}(x)}} \leq \frac{1}{\epsilon}\norm{\grad{t}{\phi_{i_0}(x)}} $ for $t\geq t_0+1$.
	Since $\lambda\leq \frac{\epsilon^2}{100dR}$, this probability is bounded by $2^{-R}$. Therefore, with probability at least $1-2^{-R}$,
	\[
	\norm{ \grad{t_0+1}{p_{i_0}(x+\sqrt{\lambda}y)}}\geq
	\norm{ \grad{t_0+1}{\phi_{i_0}(x)}} -
	\norm{ \grad{t_0+1}{p_{i_0}(x+\sqrt{\lambda}y)} - \grad{t_0+1}{\phi_{i_0}(x)}} > \frac{1}{2}\norm{ \grad{t_0+1}{\phi_{i_0}(x)}}\enspace .
	\]
	
	Thus, combining (a) and (b), we have that with probability at least $1-2\cdot 2^{-R}$,
	\[
	\norm{ \grad{t_0+1}{p_{i_0}(x+\sqrt{\lambda}y)}}> \frac{1}{2}\norm{ \grad{t_0+1}{\phi_{i_0}(x)}} > \frac{1}{4\epsilon} \norm{ \grad{t_0}{p_{i_0}(x+\sqrt{\lambda}y)}} \enspace ,
	\]
	and consequently $G(x+\sqrt{\lambda}y) = 0$. This gives us the bound on the following expectation
	\[
		\expect{y\sim\NN\br{0,1}^n}{
			F(x+\sqrt{\lambda}y)G(x+\sqrt{\lambda}y)
		} \leq 2\cdot 2^{-R}\enspace .
	\]
	Since $Y$ is $dR$-wise independent, the above argument still holds for $Y$.
	So, 
	\[
	\abs{
		\expect{Y}{
			F(x+\sqrt{\lambda}Y)G(x+\sqrt{\lambda}Y)
		}
		-
		\expect{y\sim\NN\br{0,1}^n}{
			F(x+\sqrt{\lambda}y)G(x+\sqrt{\lambda}y)
		}
	} \leq 4\cdot 2^{-R} \enspace .
	\]
	
	Now suppose that $x$ is good. That is, for any $1\leq i\leq k$ and $0\leq t\leq d-1$, $ \norm{\grad{t+1}{\phi_i(x)}} \leq \frac{1}{\epsilon}\norm{\grad{t}{\phi_i(x)}} $.
	In this case, we will prove that
	the sign of ${ p_i(x+\sqrt{\lambda}y) }$ is the same as the sign of $\phi_i(x)$ with high probability over the random variable $y$. Therefore, $F(x+\sqrt{\lambda}y)$ is almost like a constant,
	since the value of $F(x+\sqrt{\lambda}y)$ only depends on
	the signs of all ${ p_i(x+\sqrt{\lambda}y) }$.
	To show the signs of ${ p_i(x+\sqrt{\lambda}y) }$ and $\phi_i(x)$ are the same,
	it suffices to show $\frac{p_i(x+\sqrt{\lambda}y)}{\phi_i(x)} > 0 $.
	Let 
	\[
		q_i(y) \coloneq \frac{p_i(x+\sqrt{\lambda}y)}{\phi_i(x)} - 1 =
		\frac{1}{\phi_i(x)}\sum_{0< \abs{\alpha}\leq d } \frac{\partial^{\alpha}\phi_i(x)}{\sqrt{\alpha!}} \lambda^{\abs{\alpha}/2}h_{\alpha}(y) \enspace.
	\]
	Here, we expand $p_i(x+\sqrt{\lambda}y) = \phi_i(x) + \sum_{0< \abs{\alpha}\leq d } \frac{\partial^{\alpha}\phi_i(x)}{\sqrt{\alpha!}} \lambda^{\abs{\alpha}/2}h_{\alpha}(y)$
	according to \lem{expanding}.
	We have by applying hypercontractive inequality in \thm{hc}, 
	\begin{align*} 
		\norm{ q_i(y) }_R \leq  \norm{U_{\sqrt{R}} q_i(y) }_2
		&= \norm{ \frac{1}{\phi_i(x)}\sum_{0< \abs{\alpha}\leq d } \frac{\partial^{\alpha}\phi_i(x)}{\sqrt{\alpha!}}R^{\abs{\alpha}/2}\lambda^{\abs{\alpha}/2}h_{\alpha}(y) }_2\\
		&\leq \sqrt{\sum_{0< t\leq d } \frac{\norm{\grad{t}{\phi_i(x)}}^2}{\phi_i(x)^2}(\lambda R)^{t}} \leq \sqrt{\sum_{0< t\leq d } \br{\frac{\lambda R}{\epsilon^2}}^t } \enspace .
	\end{align*}
	In the last inequality, we use our assumption that
	$ \norm{\grad{t}{\phi_i(x)}} \leq \frac{1}{\epsilon}\norm{\grad{t-1}{\phi_i(x)}} \leq \cdots\leq \frac{\abs{\phi_i(x)}}{\epsilon^t}$.
	Since $\frac{\lambda R}{\epsilon^2}$ is sufficiently small, we have $\norm{ q_i(y) }_R\leq \frac{1}{4}$. Therefore, by Markov's inequality, we have
	\[
	\Pr_y \Br{ \abs{q_i(y)}  \geq \frac{1}{2}}\leq 2^{R} \cdot \norm{ q_i(y) }_R^R  \leq 2^{-R}\enspace .
	\] 
	This means that with probability at least $1-2^{-R}$,
	we have
	$ \abs{\frac{p_i(x+\sqrt{\lambda}y)}{\phi_i(x)} - 1} \leq \frac{1}{2}, $
	and therefore
	$ \frac{p_i(x+\sqrt{\lambda}y)}{\phi_i(x)} \geq \frac{1}{2} . $
	Thus, with probability at least $1-2^{-R}$, the sign of ${ p_i(x+\sqrt{\lambda}y) }$ is the same as the sign of $\phi_i(x)$. Then by a union bound,
	\[
	\Pr_y\Br{ \forall 1\leq i\leq k,\ \sign{ p_i(x+\sqrt{\lambda}y) } = \sign{\phi_i(x)} }\geq 1-k\cdot2^{-R}\enspace .
	\]
	Let $c = f\!\br{\sign{\phi_1\!(x)},\dots,\sign{\phi_k\!(x)}}$ be a constant. We have
	\[
	\abs{
		\expect{y}{
			F(x+\sqrt{\lambda}y)G(x+\sqrt{\lambda}y)
		}
		-
		\expect{y}{
			c\cdot G(x+\sqrt{\lambda}y)
		}
	} \leq k\cdot 2^{-R} \enspace .
	\]
	The same argument applying on $2dR$-wise independent Gaussian vector $Y$ gives us
	\[
	\abs{
		\expect{Y}{
			F(x+\sqrt{\lambda}Y)G(x+\sqrt{\lambda}Y)
		}
		-
		\expect{Y}{
			c\cdot G(x+\sqrt{\lambda}Y)
		}
	} \leq k\cdot 2^{-R} \enspace .
	\]
	By \lem{molli}, we know $\abs{
		\expect{Y}{
			G(x+\sqrt{\lambda}Y)
		}
		-
		\expect{y\sim\NN\br{0,1}^n}{
			G(x+\sqrt{\lambda}y)
		}
	} = kd\cdot 2^{-\Omega(R)}$.
	Therefore, we have 
	\[
	\abs{
		\expect{Y}{
			F(x+\sqrt{\lambda}Y)G(x+\sqrt{\lambda}Y)
		}
		-
		\expect{y\sim\NN\br{0,1}^n}{
			F(x+\sqrt{\lambda}y)G(x+\sqrt{\lambda}y)
		}
	} = kd2^{-\Omega(R)} \enspace .
	\]
\end{proof}


\subsection{Proof of \thm{cont}} \label{sec:main}	
\begin{proof}[Proof of \thm{cont}]
	By \lem{good_event}, the following holds with probability at least $1-\epsilon kd^3$:
	\[
		\norm{\grad{t}{p_i(y)}} \leq O\br{\frac{1}{\epsilon}} \norm{\grad{t-1}{p_i(y)}}
		\text{ for all } 1\leq t\leq d \text{ and } 1\leq i\leq k.
	\] 
	Recall the function
	$
		G(x)=\prod_{i=1}^{k} \prod_{t=0}^{d-1} \rho\!\br{\log\!\br{\frac{\norm{\grad{t}{p_i(x)}}^2}{16\epsilon^2 \norm{\grad{t+1}{p_i(x)}}^2 }}}
	$
	as defined in \eq{molli}.
	Note that $G(x) = 0$ if 
	there exists some $i$ and $t$ such that
	$\norm{\grad{t}{p_i(y)}} > O\br{\frac{1}{\epsilon}} \norm{\grad{t-1}{p_i(y)}}$.
	Thus, 
	$ \Pr_{y}[G(y) = 1] \geq 1-O(\epsilon k d^3)$.
	We have
	\begin{align*}
		\expect{y}{F(y)} &= \expect{y}{F(y)(1-G(y))}+\expect{y}{F(y)G(y)} \\
		&\leq \expect{y}{1-G(y)}+\expect{Y}{F(Y)G(Y)} + \abs{ \expect{y}{F(y)G(y)} - \expect{Y}{F(Y)G(Y)} } \\
		&\leq O(\epsilon k d^3)+\expect{Y}{F(Y)} + \abs{ \expect{y}{F(y)G(y)} - \expect{Y}{F(Y)G(Y)} } \enspace.
	\end{align*}
	Let $y = \frac{1}{\sqrt{L}}\sum_{i=1}^L y_i$ where $y_i\sim \NN(0,1)^n$ and denote $Z^i =\frac{1}{\sqrt{L}}(y_1 + \cdots +y_{i-1} +Y_{i+1} +\cdots Y_L)$.
	We have
	\begin{align*}
		&\abs{ \expect{y}{F(y)G(y)} - \expect{Y}{F(Y)G(Y)} } \\
		\leq 
		&\sum_{i=1}^{L} \abs{\expect{Z^i,Y_i}{F\!\br{Z^i + \frac{1}{\sqrt{L}}Y_i}G\!\br{Z^i + \frac{1}{\sqrt{L}}Y_i}} - \expect{Z^i,y_i}{F\!\br{Z^i + \frac{1}{\sqrt{L}}y_i}G\!\br{Z^i + \frac{1}{\sqrt{L}}y_i}} }\\
		\leq & kdL\cdot2^{-\Omega(R)}
	\end{align*}
	where the last inequality is from \lem{hybrid}.
	Thus, 
	\begin{align*}
		\expect{y}{F(y)} \leq  \expect{Y}{F(Y)} + O(\epsilon k d^3) + kdL\cdot2^{-\Omega(R)} \enspace.
	\end{align*}
	And the other side follows from considering $1-F(x)$.
\end{proof}

%% file: discrete.tex
To give an explicit construction of a PRG,
we need a discretization of $R$-wise independent Gaussian distributions. 
In this section, we show an algorithm which outputs $L$ vectors $\st{X_i}_{1\leq i\leq L}$ approximating $Y_i$, that is, $\abs{X_{i,j}-Y_{i,j}}$ is sufficiently small. Before that, we first prove that if $X$ and $Y$ are close enough, then $X$ also fools any function of low-degree polynomial threshold functions.

\begin{lemma}\label{lem:discrete}
	Let $0<\epsilon,\delta<1$, and $R\in\N$ be an integer.
	Let $Y = \frac{1}{\sqrt{L}}\sum_{i=1}^L Y_i$ where $Y_i$ is an $R$-wise independent Gaussian vector of length $n$ for $1\leq i\leq L$.
	Let $p_1,\dots,p_k:\R^n \to \R$ be arbitrary polynomials of degree $d$ and $f:\bit{k}\to\st{0,1}$ be an arbitrary Boolean function.
	Define functions
	\[
		F(x)\coloneqq f\!\br{\sign{p_1\!(x)},\dots,\sign{p_k\!(x)}}
	\]
	Suppose that for any such function $F$,
	\[
	\abs{
		\expect{Y}{
			F(Y)	
		}
		-
		\expect{y\sim\NN\br{0,1}^n}{
			F(y)
		}
	} \leq \epsilon \enspace .
	\]
	Suppose that $\st{X_{i}}_{1\leq i\leq L}$ are random vectors of length $n$ and there is a joint distribution over $X$ and $Y$ such that
	for each $1\leq i\leq L,1\leq j\leq n$, 
	$\Pr\Br{\abs{X_{i,j}-Y_{i,j}}\leq \delta}\geq 1-\delta$.
	
	Let $X = \frac{1}{\sqrt{L}}\sum_{i=1}^L X_i$ and we have that for any such function $F$
	\[
	\abs{
		\expect{X}{
			F(X)	
		}
		-
		\expect{y\sim\NN\br{0,1}^n}{
			F(y)
		}
	} \leq \epsilon + k 2^{2k}d \delta^{1/d} \sqrt{nL} \log\frac{1}{\delta}  + O(2^{2k}nL\delta) \enspace .
	\]
\end{lemma}

\begin{proof}
	Let $q_i(x) = p_i(x)+ \delta (nL)^{d/2} \br{\log\frac{1}{\delta}}^{d}$ and $\widetilde{F}(x) = f\!\br{\sign{q_1\!(x)},\dots,\sign{q_k\!(x)}}$.
	We will prove that
	\begin{itemize}
		\item[(a)] $\expect{}{
			F(X)	
		} \leq \expect{}{
			\widetilde{F}(y)	
		} + \epsilon + O(2^{2k}nL\delta)$,
		\item[(b)] $\expect{}{
			\widetilde{F}(y)	
		} \leq \expect{}{
			F(y)	
		} + k2^{2k} d \delta^{1/d} \sqrt{nL} \log\frac{1}{\delta}$.
	\end{itemize}
	Combing (a) and (b), we have
	\[
		\expect{}{
			F(X)	
		} \leq \expect{}{
			F(y)	
		} + k2^{2k} d \delta^{1/d} \sqrt{nL} \log\frac{1}{\delta} + \epsilon + O(2^{2k}nL\delta)\enspace.
	\]
	The other side can be obtained in a similar way by considering $1-F(x)$.

	\paragraph{Proving (a).}
	Since $\widetilde{F}$ is a function of degree-$d$ PTFs and $Y$ fools such a function,
	we have
	$
		\expect{}{
			\widetilde{F}(Y)	
		} \leq \expect{}{
			\widetilde{F}(y)	
		} + \epsilon .
	$
	Therefore, it suffices to prove that
	$
		\expect{}{
			F(X)	
		} \leq \expect{}{
			\widetilde{F}(Y)	
		} + O(2^{2k}nL\delta).
	$
	
	Fix a set $S\subseteq [k]$. Let $P_S(x)= \prod_{i\in S} \sign{p_i(x)}$ and $Q_S(x) = \prod_{i\in S} \sign{q_i(x)}$.
	We first show that
	\begin{align*}
		\expect{}{
			P_S(X)	
		} \leq \expect{}{
			Q_S(Y)	
		} + O(nL\delta) \enspace.
	\end{align*}
	Let event $\EE$ denote that for all $i\in [L],  j\in [n]$, $\abs{Y_{i,j}} \leq \log \frac{1}{\delta}$ and $\abs{X_{i,j}-Y_{i,j}}\leq \delta$. By the tail bound of the standard Gaussian distribution, we have $\Pr[\EE]\geq 1 - O(nL\delta)$.
	We have
	\begin{align*}
		\expect{}{
			P_S(X)	
		} = \Pr \Br{ \bigwedge_{i\in S} p_i(X)\geq 0 } &\leq \Pr \Br{ \bigwedge_{i\in S} p_i(X)\geq 0 \wedge \EE} + \Pr\Br{\widebar{\EE}}\\
		&\leq \Pr \Br{ \bigwedge_{i\in S} p_i(Y)\geq -\delta (nL)^{d/2} \br{\log\frac{1}{\delta}}^{d} } + O(nL\delta) \\
		&=\expect{}{
			Q_S(Y)	
		} + O(nL\delta) \enspace,
	\end{align*}
	where the second inequality is by \lem{close} and viewing $p_i\!\br{\frac{1}{\sqrt{L}}\sum_{i=1}^L X_i}$ as a degree $d$ function of $nL$ variables.
	
	Note that $f(x) = \sum_{S\subseteq [k] } f(\id_S) \prod_{i\in S} x_i \prod_{i\notin S}(1-x_i) $ where $\id_S$ denotes the length-$k$ string with $1$'s only at coordinates in $S$. This can be further simplified in the form
	\[
		f(x) = \sum_{S\subseteq [k] } c_S \prod_{i\in S} x_i \enspace,
	\]
	with $\sum_S \abs{c_S} \leq 2^{2k}$. So, we have
	\begin{align*}
		\expect{}{
			F(X)	
		}  = 
		\sum_{S\subseteq [k] } c_S \expect{}{
			P_S(X)	
		}
		&=
		\sum_{S\subseteq [k] } c_S \expect{}{
		 	Q_S(Y)	
		}
		+ \sum_{S\subseteq [k] } c_S \br{\expect{}{
			P_S(X)	
		} - \expect{}{
		 	Q_S(Y)	
		}}\\
		&\leq \sum_{S\subseteq [k] } c_S \expect{}{
		 	Q_S(Y)	
		} + O(2^{2k}nL\delta) 
		= \expect{}{
			\widetilde{F}(Y)	
		}+ O(2^{2k}nL\delta)\enspace.
	\end{align*}
	\vspace{-1em}
	\paragraph{Proving (b).}
	Next we show $\expect{}{
			\widetilde{F}(y)	
		}$ and $\expect{}{
			F(y)	
		}$ are close.
	Note that
	\begin{align*}
		\expect{}{
			Q_S(y)	 
		} &= \Pr \Br{ \bigwedge_{i\in S} p_i(y)\geq -\delta (nL)^{d/2} \br{\log\frac{1}{\delta}}^{d} }
		&\leq \Pr \Br{ \bigwedge_{i\in S} p_i(y)\geq 0 } + k d \delta^{1/d} \sqrt{nL} \log\frac{1}{\delta}\enspace,
	\end{align*}
	where the inequality is from \lem{anti}.
	Thus, we have
	\[
		\expect{}{
			Q_S(y)	 
		} \leq \expect{}{
			P_S(y)	 
		} + k d \delta^{1/d} \sqrt{nL} \log\frac{1}{\delta}	\enspace .
	\]
	Furthermore, we know that
	\begin{align*}
		\expect{}{
			\widetilde{F}(y)	
		}
		= \sum_{S\subseteq [k] } c_S \expect{}{
		 	Q_S(y)	
		}&= \sum_{S\subseteq [k] } c_S \expect{}{
		 	P_S(y)	
		} + \sum_{S\subseteq [k] } c_S \br{\expect{}{
		 	Q_S(y) - P_S(y)	
		}}\\
		&\leq \expect{}{
			F(y)	
		} + k2^{2k} d \delta^{1/d} \sqrt{nL} \log\frac{1}{\delta} \enspace .
	\end{align*}
	
\end{proof}


We now prove the main theorem for constructing an explicit pseudorandom generator .
The idea is that a standard Gaussian variable can be generated using two uniform $[0,1]$ random variables
through the Box–Muller transform \cite{Bm58}.
Let
$
	Y_{i,j} = \sqrt{-2 \log u_{i,j}}\cos(2\pi v_{i,j})
$ 
where $u_{i,j}$ and $v_{i,j}$ are uniform in $[0,1]$. Then $Y_{i,j}$ is a Gaussian variable.
Thus, if we truncate $u_{i,j}$ and $v_{i,j}$ to a certain precision and produce $X_{i,j}$ in a similar manner, $X$ approximates $Y$ with high probability.

\begin{theorem}\label{thm:main_formal}
	There exists an explicit PRG which $\epsilon$-fools any functions of any $k$ degree-$d$ polynomial threshold functions over $\NN(0,1)^n$ with seed length
	$
		O\br{\frac{k^5d^{11}}{\epsilon^{2}} \mathrm{log}\frac{kdn}{\epsilon} } .	
	$
\end{theorem}

\begin{proof}
	In \thm{cont}, set parameter $\epsilon$ as $\frac{\epsilon}{kd^3}$ and set $R$ as $C\log\frac{kd}{\epsilon}$ for some large constant $C$.
	Then for $L = C'\cdot \frac{k^4d^9}{\epsilon^2}\cdot \mathrm{polylog}\frac{kd}{\epsilon}$ where $C'$ is a large constant, we have
	\[
	\abs{
		\expect{Y}{
			F(Y)	
		}
		-
		\expect{y\sim\NN\br{0,1}^n}{
			F(y)
		}
	}\leq O(\epsilon) \enspace .
	\]
	
	Similar to the proof of Corollary 2 in \cite{Kan11b},
	we can let $Y_{i,j}$ generated by 
	\[
		Y_{i,j} = \sqrt{-2 \log u_{i,j}}\cos(2\pi v_{i,j})
	\]
	where $u_i = (u_{i,1},\dots,u_{i,n}) $ and $v_i= (v_{i,1},\dots,v_{i,n})$ are $2dR$-wise independent uniform $[0,1]$ random vectors.
	Then let $u'_{i,j}$ and $v'_{i,j}$ be $M = C''kd\log\frac{kdn}{\epsilon}$-bit approximations to  $u_{i,j}$ and $v_{i,j}$
	(i.e., round $u_{i,j}$ and $v_{i,j}$ to multipels of $2^{-M}$), where $C''$ is a large constant, and let
	\[
		X_{i,j} = \sqrt{-2 \log u'_{i,j}}\cos(2\pi v'_{i,j})\enspace .
	\]
	Letting $\delta = \Omega(2^{-M/2})$, for the same reason as in the proof of Corollary 2 in \cite{Kan11b}, we have
	$\abs{X_{i,j}-Y_{i,j}}<\delta$ with probability at least $1-\delta$.
	Then, by \lem{discrete},
	\[
	\abs{
		\expect{X}{
			F(X)	
		}
		-
		\expect{y\sim\NN\br{0,1}^n}{
			F(y)
		}
	}\leq O(\epsilon) + k 2^{2k}d \delta^{1/d} \sqrt{nL} \log\frac{1}{\delta}  + O(2^{2k}nL\delta)  = O(\epsilon)\enspace .
	\]

	Note that $X_i$ can be generated by $2dR$-wise independent random variables $u'_{i,j}$ and $v'_{i,j}$ taken uniformly from $\st{2^{-M}, 2\cdot 2^{-M}, 3\cdot 2^{-M},\dots,1}$ using $O(dRM)$ randomness.
	Thus generating $X$ uses $O(LdRM) =
	O\br{\frac{k^5d^{11}}{\epsilon^{2}} \mathrm{log}\frac{kdn}{\epsilon} }$ randomness.
	
	
%
\end{proof}

%% file: app.tex
\section{Facts about Bump Function}
\label{sec:fact_bump}

\begin{repfact}{fact:deri1}
		For all $t\in \N$, $\abs{\Psi^{(t)}(x) } \leq t^{(3+o(1))t}$.
\end{repfact}

\begin{proof}
	It is easy to check there exists a series of polynomials $\st{P_t}_{t\in\N}$ such that for $x\in(-1,1)$
	\[
		\Psi^{(t)}(x) = \frac{P_t(x)}{(1-x^2)^{2t}}\cdot \Psi(x)\enspace.
	\]
	Besides, $\st{P_t}_{t\in\N}$ has the following recursion:
	\[
		P_0(x) = 1, \enspace P_1(x)=-2x,\enspace P_{t} = (1-x^2)^2P'_{t-1}(x)+4(t-1)x(1-x^2)P_{t-1}(x)-2xP_{t-1}(x)\enspace.
	\]
	The degree of $P_t$ is at most $3t$.
	Therefore we have
	\[
		\abs{\Psi^{(t)}(x) } \leq \br{\max_{x\in(-1,1)} \abs{P_t(x)}} \cdot
		\br{\max_{x\in(-1,1)} \frac{\Psi(x)}{(1-x^2)^{2t}}}\enspace.
	\]
	Let $f(x) = x\cdot e^{-\frac{x}{2t}}$ for $x\in(1,+\infty)$. We have $f(x)\leq f(2t) = \frac{2t}{e}$ by a simple calculation.
	Thus, we have that $\frac{\Psi(x)}{(1-x^2)^{2t}}\leq \br{\frac{2t}{e}}^{2t}$. We are left to bound $\max_{x\in(-1,1)} \abs{P_t(x)}$.
	
	Define $\norm{P}_1$ be the sum of absolute values of all coefficients of the polynomial $P$.
	Since $x\in(-1,1)$, we know $\max_{x\in(-1,1)} \abs{P_t(x)}\leq \norm{P_t}_1$.
	By the recursion, we have
	\[
		\norm{P_t}_1 \leq 4\norm{P_{t-1}'}_1 + 8t\norm{P_{t-1}}_1 \leq 20t\norm{P_{t-1}}_1
	\]
	where the last inequality is from $\norm{P_{t-1}'}_1\leq 3t\norm{P_{t-1}}_1$ since the degree of $P_{t-1}$ is at most $3t-3$.
	So we know
	$
		\max_{x\in(-1,1)} \leq 20^t\cdot t!.
	$
	Therefore,
	\[
		\abs{\Psi^{(t)}(x) } \leq \br{\frac{2t}{e}}^{2t} \cdot 20^t\cdot t! = t^{(3+o(1))t}\enspace.
	\]
\end{proof}

\begin{repfact}{fact:deri2}
	For all $t\in \N$, $\abs{\rho^{(t)}(x)} \leq t^{(3+o(1))t}$.
\end{repfact}

\begin{proof}
	Directly from \fct{deri1} by observing that
	$\rho(x) =e\cdot \Psi(1-x) $ for $ 0<x<1$ and is constant elsewhere.
\end{proof}

\begin{repfact}{fact:deri}
	Let $r(u,v) \coloneq \rho( \log u-\log v + c )$ for some constant $c$. Then we have that for all $n,m\in \N$,
	$ \abs{ \frac{ \partial^n \partial^m r(u,v) }{ \partial u^n \partial v^m } } \leq \frac{(n+m)^{6(n+m)}}{\abs{u}^n \abs{v}^m} $.
\end{repfact}

\begin{proof}
	Let $g(u,v) = \log u-\log v + c$. Then by the generalized chain rule for the derivative of the composition of two functions (also known as Fa\`a di Bruno's formula), we have
	\begin{align*}
		\frac{ \partial^n \partial^m r(u,v) }{ \partial u^n \partial v^m }
		= &\sum_{\substack{(a_1,\dots,a_n)\in\N^n\\ a_1+2\cdot a_2+\cdots +n\cdot a_n=n }}
		\sum_{\substack{(b_1,\dots,b_m)\in\N^m\\ b_1+2\cdot b_2+\cdots +m\cdot b_m=m}} \frac{n!m!}{\prod_{i=1}^n \br{i!}^{a_i}a_i! \prod_{i=1}^m \br{i!}^{b_i}b_i!}
		\\ &\hspace{2cm}\cdot \rho^{(a_1+\cdots+a_n+b_1+\cdots+b_m)}(g(u,v))
		\cdot\prod_{i=1}^n \br{\frac{(-1)^ii!}{u^i}}^{a_i}
		\prod_{i=1}^m \br{\frac{(-1)^{i+1}i!}{v^i}}^{b_i}\enspace.
	\end{align*}
	Therefore
	\begin{align*}
		&\abs{ \frac{ \partial^n \partial^m r(u,v) }{ \partial u^n \partial v^m } } \leq n^n\cdot m^m\cdot n!\cdot m!\cdot (m+n)^{4(m+n)}\cdot \frac{1}{\abs{u}^n\abs{v}^m} \leq \frac{(n+m)^{6(n+m)}}{\abs{u}^n \abs{v}^m}\enspace.
	\end{align*}
\end{proof}